\documentclass[12pt,a4paper]{article}
\pdfoutput=1
\usepackage{hyperref}
\usepackage{graphicx}
\usepackage{amssymb,amsfonts,amsthm}
\usepackage{amsmath}
\usepackage[makeroom]{cancel}
\usepackage{mathtools}
\usepackage{mathrsfs,pifont}
\usepackage{amsfonts}
\usepackage[matrix,arrow,curve]{xy}
\usepackage[english]{babel}
\usepackage[body={14.5cm, 20.5cm},left=3cm, top=3cm]{geometry}

\usepackage{accents}

\usepackage{epsf}
\usepackage{graphicx}
\usepackage[font={small,sl}]{caption}
\def\be{\begin{eqnarray*}}
\def\ee{\end{eqnarray*}}
\def\bew{\begin{eqnarray*}}
\def\eew{\end{eqnarray*}}

\def\p{\partial}
\def\tr{{\rm tr}\,}

\def\l[{\phantom.[}

\newlength{\dhatheight}

\newcommand{\C}{\mathbb{C}}
\newcommand{\Z}{\mathbb{Z}}
\newcommand{\Ct}{\C^\times}
\newcommand{\bP}{\mathbb{P}}
\newcommand{\cO}{\mathscr{O}}
\newcommand{\cK}{\mathscr{K}}
\newcommand{\cZ}{\mathscr{Z}}
\newcommand{\cT}{\mathscr{T}}
\newcommand{\cF}{\mathscr{F}}
\newcommand{\cG}{\mathscr{G}}
\newcommand{\tO}{\widehat{\cO}}
\newcommand{\vir}{\textup{vir}}
\newcommand{\bT}{\mathsf{T}}
\newcommand{\bTt}{\widetilde{\bT}}
\newcommand{\cE}{\mathscr{E}}
\newcommand{\cA}{\mathscr{A}}
\newcommand{\cV}{\mathscr{V}}

\newcommand{\bSd}{{\mathsf{S}}^{\raisebox{0.5mm}{$\scriptscriptstyle
      \bullet$}}}
\newcommand{\sSd}{{\mathbb{S}}^{\raisebox{0.5mm}{$\scriptscriptstyle
      \bullet$}}}

\DeclareMathOperator\PT{PT}
\DeclareMathOperator\Coker{Coker}
\DeclareMathOperator\rk{rk}
\DeclareMathOperator\pt{pt}
\DeclareMathOperator\Coh{Coh}
\DeclareMathOperator\Ind{Ind}
\DeclareMathOperator\Aut{Aut}
\DeclareMathOperator\Attr{Attr}

\theoremstyle{theorem}
\newtheorem{thm}{Theorem}
\newtheorem{lem}{Lemma}[section]
\newtheorem{prop}[lem]{Proposition}
\newtheorem{cor}[lem]{Corollary}

\theoremstyle{definition}

\theoremstyle{remark}

\DeclareMathOperator\Ker{Ker}

\DeclareMathOperator\ev{ev}
\DeclareMathOperator\Hilb{Hilb}
\DeclareMathOperator\Stab{Stab}
\DeclareMathOperator\Fr{Fr}
\DeclareMathOperator\supp{supp}

\let\emptyset\varnothing

\title{The 2-leg vertex in K-theoretic DT theory}
\author{Ya.~Kononov, A.~Okounkov, A.~Osinenko}

\begin{document}

\maketitle

\begin{abstract}
K-theoretic Donaldson-Thomas counts of curves in toric and many
related threefolds can be computed in terms of a certain canonical
3-valent tensor, the K-theoretic equivariant vertex. In this paper we
derive a formula for 
the vertex in the case when two out of three entries are
nontrivial. We also discuss some applications of this result.  
\end{abstract}

\section{Introduction}

\subsection{}

Donaldson-Thomas (DT) theories, broadly interpreted, are enumerative
theories 
of objects  that look like coherent sheaves on a algebraic threefold. In
this very broad spectrum of possibilities, DT counts of curves in
an algebraic threefold $X$ stand out due to the intrinsic richness
of the subject, of the generality in which such counts may be defined
and studies, and also because of the range of connections with other
branches of mathematics and mathematical physics. See, for example,  
\cite{OkTak} for a set of introductory lectures, and also 
\cite{INOV} for an early discussion of the meaning of DT counts
in theoretical physics. From the perspective of both mathematics and physics, it is particularly natural to study DT counts in
equivariant K-theory, which is the setting of this paper. 

In contrast to counts defined only with assumptions like 
$c_1(X)=0$, DT counts of curves in general 3-folds $X$ are much more flexible. The degeneration
and localization properties of these counts (see \cite{OkTak} for an
introduction), make the theory resemble the 
computation of  Chern-Simons (CS) counts for
real 3-folds by cutting and
gluing. Similarly to how CS counts may be reduced to a few basic
tensors (described in terms of quantum groups), there are some basic
tensors for DT counts of curves, of which the 3-valent K-theoretic
vertex is the most important one. 

\subsection{}

The 3-valent vertex is defined as the equivariant count of curves in
the coordinate space $X=\C^3$. This can be defined as either straight
equivariant localization counts for suitable moduli spaces of 
one-dimensional sheaves on $\C^3$, or with
relative boundary condition along divisors $D_1,D_2,D_3$ that compactify
$X$ in some ambient geometry like $(\mathbb{P}^1)^3$. In either case, 
the vertex takes 3 partitions or, more canonically, a triple of 
elements of $K_\textup{eq}(\Hilb(\C^2,\textup{points}))$ as its
argument. Variations in boundary condition result in gauge
transformations of the vertex that are understood, albeit
complicated. In this paper, we find a particular gauge, that is, 
a particular relative geometry that makes the 2-leg vertex simple. 

\begin{figure}[!htbp]
  \centering
   \includegraphics[scale=0.84]{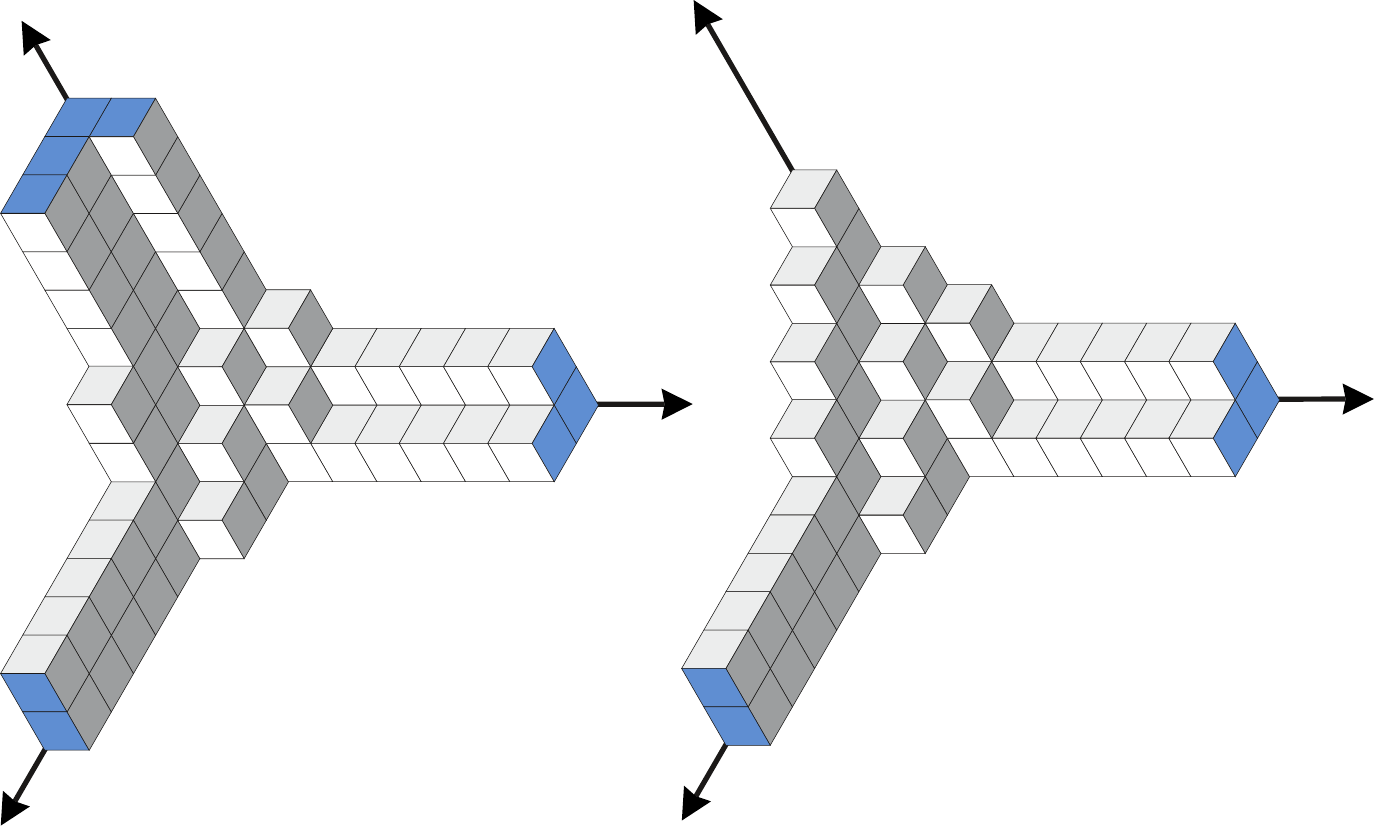}
 \caption{All torus-invariant subschemes of $\C^3$ are asymptotic to
   certain torus-invariant subschemes of $\C^2$, that is, to certain
   partitions, along the coordinate axes. These partitions are shown in blue in
   the figure. When one of them is empty, as in example on
 the right, one talks about a 2-leg vertex.}
  \label{f_vertex}
\end{figure}

Further technical variations of vertices in DT theory come from the
possibility to vary stability conditions for DT moduli spaces. While
early 
papers used the Hilbert scheme of curves in $\C^3$ to define the
vertex, there are many technical advantages to using the
Pandharipande-Thomas (PT) moduli spaces instead \cite{PandharipandeThomas}. Wall-crossing
between different stability chambers have not really been explored in
fully 
equivariant K-theory. However, there is little doubt that Hilbert
scheme and the PT counts in any $X$ differ by an overall factor that comes from
counting $0$-dimensional subschemes in $X$, see \cite{OkLect} for discussion of the
latter count. In this paper, we work with the PT counts. Their only
disadvantage is that they are harder to visualize, which is why
Figure \ref{f_vertex} shows examples of torus-fixed points in the
Hilbert scheme of curves in $\C^3$. 

\subsection{}

While there is an 
in principle understanding of the vertices in terms of the 
Fock space representations of quantum
\emph{double} loop groups (see \cite{OkTak} for an introduction), having a better handle on them would lead to
 a significant theoretical and computational progress. The goal of
this paper is to provide a direct and explicit description of the
vertex with 2 nontrivial legs (as in Figure \ref{f_vertex} on the
right) in a specific gauge. This result is stated as Theorem \ref{t_main} below. 

Given the complexity of the problem,
we find the existence of such an explicit formula quite remarkable. We
also think it is unlikely that a 
comparably direct formula exists for the full 3-valent vertex. 

\subsection{} 

The shape of our formula definitely suggests an interpretation in terms of
counting M2-branes of the M-theory, along the lines explored in
\cite{NO}. Very visibly, \eqref{mainTh} is made up of the contributions of
the three basic curves in the geometry: the two coordinate 
axes and their union. 

We note, however, that formula \eqref{mainTh} refers to relative DT counts
and those currently fall outside of the scope of the conjectural
correspondence with membrane counts proposed in \cite{NO}. 
Thus, the framework of \cite{NO} needs to be expanded and we hope to return to this question in a future paper.

\subsection{} 

As an application of our result, we compute the operator corresponding to the parallel legs in the resolved conifold. It proves that any matrix element of the operator divided by the vacuum matrix element is polynomial in the K\"ahler parameter.

\subsection{}

We thank the Simons Foundation for financial support within the Simons
Investigator program. The results of Section 3 were obtained under
support of the Russian Science Foundation under grant  19-11-00275. 

We thank N.~Nekrasov and A.~Smirnov for valuable discussions.

\section{The setup and the formula}





\subsection{The basic geometry}

\subsubsection{}

The equivariant vertex with 2 legs can be captured by relative
counts in the following threefold
\be
X = S \times \C, \ \  \text{where} \ \ S = \text{Blow-up}_{(0,0)}(\mathbb{P}^1 \times \C)
\ee
The toric diagram of $S$ is drawn on the left in Figure
\ref{f_SX}. The torus
\be
\bT = \Ct_{x} \times \Ct_{y} \times \Ct_{z}
\ee
acts on $X$ with weights as in Figure \ref{f_SX}. We
denote
by 
\be
D_i \cong \C_{y} \times \C_{x z} \,, \quad i=1,2 \,, 
\ee
the two divisors shaded in Figure \ref{f_SX}.
\begin{figure}[!htbp]
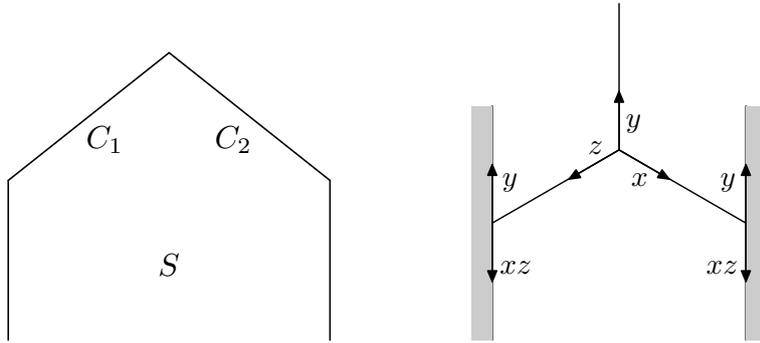

  \centering
  \includegraphics[scale=1.2]{pic5.5} \qquad \qquad
  \includegraphics[scale=1.1]{2l.1}
 \caption{On the left, the toric diagram of the surface $S$ showing the 
complete curves $C_1,C_2\subset S$. On the right, the weights of the torus action on $X$. The shading marks
   the relative divisors $D_1,D_2\subset X$.}
\label{f_SX}
\end{figure}

In computations we use the following notation: for any weight $\mathsf{w}:\bT \to \Ct$ we define
$$
\{\mathsf{w}\}  = \widehat{\mathsf{a}}(\mathsf{w}) = \mathsf{w}^{1/2} - \mathsf{w}^{-1/2},
$$
and extend it multiplicatively to linear combinations.

\subsubsection{}

The Pandharipande-Thomas moduli spaces of stable pairs parametrize
complexes of the form
\begin{equation}
\cO_X  \xrightarrow{\,\, s\,\,} \cF\label{sPT}
\end{equation}
in which $\cF$ is a pure $1$-dimensional sheaf on $X$ and the cokernel
of the section $s$ is a $0$-dimensional sheaf.

If a smooth divisor $D\subset X$ is given, there is a very useful
\emph{relative} modification $\PT(X/D)$ of this moduli space. It parametrizes
complexes of the form \eqref{sPT} on semistable degenerations $X'$ that
allow $X$ to bubble off copies of $\bP(\cO_{D_i} \oplus
N_{X/D_i})$, where $D_i$ is a component of $D$. One get such degenerations in families by blowing up 
$D_i\times \{b\}\subset
X \times B$ in a
trivial family with base $B \owns b$. See e.g.\ \cite{OkLect} for a hands-on
introduction. By allowing degenerations of $X$, one can achieve that
$\Coker(s)$ is supported away from $D$. Therefore, there
is a well-defined map
\begin{equation}
\ev: \PT(X/D) \to \Hilb(D,\textup{points})\label{ev}
\end{equation}
that takes a complex of the form \eqref{sPT} to its
restriction
\begin{equation}
\cO_D  \to \cF_D \overset{\textup{\tiny{def}}}= \cF \otimes \cO_D \to 0 \label{cFD} 
\end{equation}
to the divisor $D$. 
This map records the intersection of curves in $X$ with the divisor $D$. 

In our case, $D=D_1 \sqcup D_2$ has two components. We denote the
components of the evaluation map \eqref{ev} by 
\be
\ev_i: \PT(X/D_1 \sqcup D_2) \to \text{Hilb}(D_i)  \,. 
\ee

\subsubsection{}
The general formalism of perfect obstruction theories gives the PT
moduli spaces their virtual structure sheaves $\cO_\vir$. A small, but
important detail in setting up the K-theoretic DT counts is to use 
a certain \emph{symmetrized} virtual structure sheaf $\tO_\vir$. The
main difference between $\tO_\vir$ and $\cO_\vir$ is a twist by a
square root of the virtual canonical bundle $\cK_\vir$, in parallel to
how a Dirac operator on a K\"ahler manifold $M$ is obtained
from the $\overline{\partial}$-operator in $\Omega^{0,\bullet}(M)$
using 
a twist by $\cK_M^{1/2}$.

\subsubsection{}

The deformation theory of sheaves on a fixed semistable
degeneration $X'$ gives $\PT(X/D)$
a relative obstruction theory over the stack of 
degenerations of $X$. We denote by
\begin{equation}
  \label{Tvir}
  \cT_\vir = \chi_{X'}(\cF) + \chi_{X'}(\cF,\cO_{X'}) -
  \chi_{X'}(\cF,\cF) 
\end{equation}
and $\cK_\vir=\det(\cT_\vir)^{-1}$ the
virtual tangent bundle and the virtual
canonical line bundle of this relative obstruction theory. 

While the general discussion of \cite{NO} about the existence of square
roots may be adopted to the setting at hand, there is a more direct
argument that applies whenever the pair $(X,D)$ is a line bunde over a
surface $(S,\partial S)$, where $\partial S = S \cap D$.

The required
twist will also involve the pull-back
\begin{equation}
  \label{THilb}
  \cT_{\Hilb(D)} = \chi(\cF_D) + \chi_{D}(\cF_D,\cO_{D}) -
  \chi_{D}(\cF_D,\cF_D) 
\end{equation}
of the tangent bundle to $\Hilb(D)$ under the evaluation map
\eqref{ev}. Since $D$ is the anticanonical divisor of $X$, the
virtual dimension of the source in \eqref{ev} is half of the
dimension of the target, and in fact the image is a virtual
Lagrangian. Reflecting this, we will see a
\emph{polarization}, that is, a certain half of the tangent
bundle \eqref{THilb}
in the sense of \eqref{polarHilb} 
in the formulas. 

\subsubsection{}

For simplicity of notation, assume that $X'=X$ and let  $p: X \to S$ denote the projection. A sheaf $\cF$ on $X$ is the
same as its pushforward $p_* \cF$ together with an endomorphism of
$p_*\cF$ given by
the multiplication by the 3rd coordinate, thus
\begin{equation}
\chi_{X}(\cF,\cF) = (1-y) \, \chi_{S}(p_* \cF,p_* \cF)
\,. \label{chi1y}
\end{equation}
We define
\begin{align}
  \cT^{1/2}_{\Hilb(D)} &= \chi(\cF_D) - \chi_{\partial S}(p_* \cF_D,
                         p_*\cF_D)  \,, \label{T12H}\\
  \label{T12vir}
  \cT^{1/2}_\vir &= \chi(\cF(-D)) - \chi_{S}(p_*\cF,p_*\cF(-D))
  \,.
\end{align}
These are virtual bundles of ranks
\begin{align}
\rk  \cT^{1/2}_{\Hilb(D)} &= ([\cF],D)_X = \textup{vir
                            dim}  \, \PT(X) \,,
  \\
  \rk  \cT^{1/2}_\vir & =
  \chi(\cF) - ([\cF],D)_X + (p_*[\cF],p_*[\cF])_S \,, 
\end{align}
where $(\,\cdot\,,\,\cdot\,)_S$ denotes the intersection form on curve
classes in the surface $S$.

The first half of the following proposition shows that \eqref{T12H} is
a polarization of $\Hilb(D)$, that is, an equivariant half of the
tangent bundle. 

\begin{prop}  
  \begin{align}
    \label{polarHilb}
    \cT_{\Hilb(D)} &= \cT^{1/2}_{\Hilb(D)}  + xyz 
\left(\cT^{1/2}_{\Hilb(D)}\right)^\vee \\ 
    \cT_\vir &= \cT^{1/2}_\vir - xyz
               \left(\cT^{1/2}_{\vir}\right)^\vee + \cT^{1/2}_{\Hilb(D)}\,.
               \label{cTvir_sum}
  \end{align}
\end{prop}

\begin{proof}
Follows from \eqref{chi1y}, Serre duality, and the equivariant
identifications
$$
\cK_X = \frac{1}{xyz} \, \cO_X(-D) \,, \quad \cK_S = \frac{1}{xz} \,
\cO_S(-\partial S) \,.
$$
\end{proof}

\subsubsection{}
To have the required twist defined in equivariant K-theory, we pass to
the cover $\bTt$ of the torus $\bT$ with characters
$$
\kappa = \sqrt{xyz} \ \ \textup{and}  \ \ \sqrt{y}\,. 
$$
With this we define
\begin{equation}
\tO_\vir = (-q)^{\chi(\cF)} \, \kappa^{\rk  \cT^{1/2}_\vir } \,
\det \left(\cT^{1/2}_\vir\right)^{-1} \otimes (\sqrt{xz})^{-|\cF|_{D_1}|^2 - |\cF|_{D_2}|^2} \otimes \cO_\vir
\,, \label{tOvir} 
\end{equation}
in other terms
\begin{equation}
\tO_\vir = (-q)^{\chi(\cF)} \otimes \left(
\cK \otimes \left(\det \text{H}^\bullet( \cF|_D) \right)
\right)^{1/2}
 \otimes \cO_\vir,
\end{equation}
where $q$ is the boxcounting variable in the DT theory. 

Our main object of study is the correspondence in defined by
\begin{equation}
  \label{eq:1}
   \ev_* \tO_\vir \in K_\textup{eq}(\Hilb(D))[[q]] \,. 
\end{equation}

\subsubsection{Example}

Consider sheaves in the class $d \cO_{C_1}$ with the smallest possible euler characteristic $\chi = d$. 
The moduli space is isomorphic to $S^d \mathbb{A}^1 = \mathbb{A}^d$. There is one fixed point $\cE_d$ for each $d$, and  we can compute the characters of the tangent space to the Hilbert scheme:
$$
\left.\cT_{\Hilb(D)}\right|_{\cE_d} = \left.\cT_{S^d \mathbb{A}^1}\right|_{\cE_d} + xyz \left.\cT_{S^d \mathbb{A}^1}^\vee \right|_{\cE_d} = \sum_{i=1}^d \left( y^i + \frac{x y z}{y^i} \right) ,
$$
The normal bundle to $\mathbb{P}^1 \subset X$ with support $C_1$ is $\cO \oplus \cO(-1)$, and the second bundle is acyclic, that is why
$$ 
\left.\cT_\vir\right|_{\cE_d} = \sum_{i=1}^d y^i.
$$
By (\ref{T12H}) and (\ref{T12vir}),
$$
\left.\cT^{1/2}_{\Hilb(D)}\right|_{\cE_d} = \frac{1-y^{-d}}{1-y^{-1}} \left(
1-\frac{1-y^d}{1-y}(1-xz),
\right)
$$
$$
\left.\cT^{1/2}_\vir \right|_{\cE_d} =  (-xz) \frac{\{y^d\}^2}{\{y\}^2},
$$
$$
\left.\rk \cT_\vir^{1/2}\right|_{\cE_d} = -d^2.
$$
For the stalk of the symmetrized structure sheaf at $\cE_d$ we get
$$
\left.\tO_\vir\right|_{\cE_d} = y^{-d^2/2}
$$
in full agreement with
$$
(\cK \otimes \cO(1)_{\Hilb(D_1)})^{1/2}.
$$

\subsubsection{Example}

Consider sheaves in the class $d \cO_{C1 \cup C2}$ with the smallest possible $\chi = d$. These sheaves are obtained as pullbacks of sheaves on $\mathbb{P}^1 \times \C^2$, where we consider the component of degree $d$ and $\chi=d$.
Hence, the moduli space in this case is isomorphic to the diagonal in the product of two Hilbert scheme of $d$ points in $\C^2$.
The fixed points can be identified with Young diagrams $\lambda = (\lambda_0 \geq \lambda_1 \geq ...)$ which we denote by by $\cE_\lambda$. Characters of the tangent space is given by the well-known arms-legs formula
$$
\left.\cT_\vir\right|_{\cE_\lambda} = \sum_{\square \in \lambda} \left(
(xz)^{-a(\square)} y^{l(\square)+1} + (xz)^{a(\square)+1} y^{-l(\square)}
\right),
$$
$$
\left.\cT_{\Hilb(D)}\right|_{\cE_\lambda} = 2\cdot\sum_{\square \in \lambda} \left(
(xz)^{-a(\square)} y^{l(\square)+1} + (xz)^{a(\square)+1} y^{-l(\square)}
\right).
$$
Here we have a factor of 2 because of 2 copies of $\C^2$.
By (\ref{T12H}) and (\ref{T12vir}),
$$
\left.\cT^{1/2}_{\vir}\right|_{\cE_\lambda} = -\frac{1}{2}\cdot\cT^{1/2}_{\Hilb(D)}(\cE_\lambda),
$$
$$
\left.\rk \cT^{1/2}_\vir\right|_{\cE_\lambda} = -|\lambda|.
$$
The stalk of the symmetrized structure sheaf is equal to
$$
\left.\tO_\vir\right|_{\cE_\lambda} = {\kappa^{-|\lambda|}} \cdot \prod_\square (xz)^{-a'(\lambda)} y^{-l'(\square)} = \left.{\kappa^{-|\lambda|}} \cdot  \cO(1)_{\Hilb(D)}\right|_{\cE_\lambda}.
$$

\subsection{Symmetric functions}

\subsubsection{}

Theorems of Bridgland-King-Reid \cite{BKR} and Haiman \cite{Haiman} give an
equivalence
\begin{equation}
D^b \Coh_\bT \Hilb(\C^2,n) \cong D^b \Coh_{S(n) \times \bT}
\C^{2n} \label{BKRH}
\end{equation}
and hence a natural identification of equivariant $K$-theories. The
Fourier-Mukai kernel of this identification is $\cO_\cZ$, where $\cZ$
is the universal subscheme 
$$
\mathscr{Z} \subset \Hilb(\C^2,n) \times \C^{2n}
$$
provided by Haiman's identification of $\Hilb(\C^2,n)$ with the
Hilbert scheme of regular orbits for the diagonal action of $S(n)$ on
$\C^{2n}$.

\subsubsection{}

We use the identification
$$
K_{S(n) \times \bT} (\C^{2n})_\textup{compactly supported} \xrightarrow{\,\,\sim\,\,} K_{S(n) \times
  \bT} (\pt) \cong \Lambda_n \otimes K_\bT(\pt).
$$
given by the global sections and extend it to the suitable
localization of the right-hand side for all sheaves. Here
$\Lambda_n$ denotes symmetric functions of degree $n$.

The identification of $K_{S(n)}(\pt)$ with symmetric functions sends a
module
$W$ to the symmetric function $f_W$ such that
\begin{equation}
(f_W, p_\mu) = \tr_W \, \sigma_\mu \,,\label{fW}
\end{equation}
where
$$
\sigma_\mu = 
\textup{permutation of cycle type $\mu$} \,, 
$$
the functions $p_\mu = \prod p_{\mu_i}$ are the power-sum symmetric
functions, 
and the inner product is the standard inner product on $\Lambda_n$. 
In particular, the sheaf
$$
W^\lambda \otimes \cO_0 \in K_{S(n) \times \bT} (\C^{2n})\,, 
$$
where $W^\lambda$ is an irreducible $S(n)$-module labeled by a diagram
$\lambda$ of size $n$, corresponds to the Schur function $s_\lambda\in \Lambda_n$. 

\subsubsection{}\label{s_ab}

Let $\bT$ acts on $\C^2$ by 
$$
\begin{pmatrix}
  t_1 & \\
  & t_2 
\end{pmatrix} \cdot
\begin{pmatrix}
  a \\
  b
\end{pmatrix} = \begin{pmatrix}
  t_1 a \\
  t_2 b
\end{pmatrix}
$$
and let
$$
L = \{ b_1=\dots=b_n=0\} \subset \C^{2n}
$$
be the half-dimensional subspace defined
by the vanishing of the coordinates
of
weight $t_2$. We have
$$
\tr_{\cO_{L}}  \sigma_\mu = \prod_i \frac 1{1-t_1^{-\mu_i}}
$$
and therefore for any $S(n)$-module $W$
\begin{equation}
f_{W\otimes \cO_0} =  f_{W \otimes \cO_L} \Big|_{p_k \mapsto
  (1-t_1^{-k}) p_k} \label{e_plet} \,. 
\end{equation}
Substitutions of the kind \eqref{e_plet} are known in the literature
as plethystic substitutions.

\subsubsection{}\label{s_Ed}

Induction and restriction of representations give symmetric functions 
$$
\Lambda = \bigoplus_{n=0}^\infty \Lambda_n
$$
their natural multiplication and comultiplication, see
\cite{Zel}.

It is possible to produce correspondences between Hilbert
schemes of different size that act as multiplication by a certain
symmetric function.
Concretely, let $b$ be the coordinate of weight $t_2$ and consider the
correspondence 
\be
E_d = \{ I_1 \subset I_2 |  \ b \cdot (I_2/I_1) = 0, \, \dim(I_2/I_1)
= d \} \subset \text{Hilb}^* \times \text{Hilb}^{*+d} \,, 
\ee
which can be shown to be smooth and Lagrangian. The following
Proposition may be deduced e.g.\ from the computations in \cite{Negut}.

\begin{prop}\label{p_ed} 
 The correspondence $E_d$ acts as multiplication by a symmetric
 function. 
\end{prop}

\subsubsection{}
To determine the symmetric function in \eqref{p_ed}, it suffices
to apply the correspondence to the Hilbert scheme of $0$ points, in
which case we get structure sheaf $\cO_L$ as the corresponding
element of $K_{S(n) \times \bT} (\C^{2n})$. This corresponds to the
trivial module $W=\C$ in \eqref{e_plet}, and thus to the complete
homogeneous symmetric function $h_n = s_{(n)}$. From the
generating function
$$
\sum_n h_n(x) = \prod_i \frac1{1-x_i}
$$
we conclude that $E_d$ multiplies by the degree $d$ component of
\begin{multline}
  \label{HHd}
  \prod_i \prod_{m\ge 0} \frac1{(1-t_1^{-m} x_i)} =
  \exp \left(\sum \frac{p_k} 
    {k(1-t_1^{-k})} \right) = \\
  = \sum \frac{H_{(d)}}{(1-t_1^{-1}) \dots (1-t_1^{-d})} \,. 
\end{multline}
Here $H_{(d)}$ is the Macdonald polynomial in Haiman's normalization,
it corresponds to the unique fixed points in
\begin{equation}
  L/S(n) \cong \Hilb(\C^1,n) \cong \C^n
  \subset \Hilb(\C^2,n) \label{imL}
\end{equation}
The
denominator is given by the tanget weights to \eqref{imL}. The
correspondence $E_d$ maps to \eqref{imL} by the class of $I_2/I_1$
and, in principle, we can pull back the point class instead of the
structure sheaf from there. That would give multiplication by
$H_{(d)}$.


\subsection{Symmetric algebras}

\subsubsection{}

Let $V$ be a representation of a group $G$. The symmetric algebra
of $V$
$$
\bSd V = \sum_k \left(V^{\otimes k}\right)^{S(k)}
$$
is a representation of $V$ with character
\begin{equation}
\tr_{\bSd V} \, g = \exp\left( \sum_n \frac1n \tr_V \, g^n 
\right)\label{Sdg} \,. 
\end{equation}
Because of the relation
$$
\bSd (V_1 \oplus V_2) = \bSd V_1 \, \otimes \bSd V_2
$$
it is enough to check \eqref{Sdg} for a $1$-dimensional module of some
weight $w$, in which case it gives
$$
\frac{1}{1-w} = \exp\left(\sum_n \frac1n w^n \right) \, .
$$
The operation $\bSd$ extends naturally to $K_G(\pt)$, that is, to
virtual representations of $G$, by
the rule
$$
\bSd (-V) = \sum_k (-1)^k {\bigwedge}^{\!\!k} \, V \,.
$$

\subsubsection{}
Now suppose $V$ is direct sum of representations of $G$ with action of
some symmetric group $S(k)$, $k\ge 1$, that is, 
$$
V = \sum V_k \in \bigoplus_{k=1}^\infty
K_{G\times S(k)}(\pt) \subset K_G(\pt)\otimes \Lambda \,. 
$$
We define
\begin{equation}
\sSd  V  = \bigoplus_{\mathbf{k}=\{k_1,k_2, \dots\} } \Ind^{S(\sum
  k_i)}_{\Aut(\mathbf{k}) \ltimes \prod S(k_i)} \bigotimes_{k_i\in \mathbf{k}}
V_{k_i}  \,.\label{sSdd}
\end{equation}
where the sum is over all multisubsets, that is, subsets with
repetitions, $\mathbf{k}$ of $\{1,2,\dots\}$ and the group
$\Aut(\mathbf{k})$ permutes equal parts of $\mathbf{k}$.
The representation \eqref{sSdd} is the only natural representation
that can be made out of unordered collections of representations
$V_{k_i}$.

\subsubsection{}
For any pair of groups $H\subset G$, Frobenius reciprocity implies
$$
\left(\Ind^G_H V\right)^G =V^H \,.
$$
This yields the following compatibility between $\bSd$ and $\sSd$
\begin{equation}
  \xymatrix{
    K_G(\pt)\otimes \Lambda \ar[dd]_{\sSd}
    \ar[rrr]^{\textup{invariants}}  &&&K_G(\pt) \ar[dd]^{\bSd} \\ \\
    \overline{K_G(\pt)\otimes \Lambda} \ar[rrr]^{\textup{invariants}}
    &&&
    \overline{K_G(\pt)} 
}\,. \label{SVSV}
\end{equation}
The bars in the bottom line of the diagram \eqref{SVSV}
stand for the required completions.

\subsubsection{}

Take $f(g) \in K_G(\pt)\otimes \Lambda$, where the argument $g$
denotes an element of the group $G$.  Define the Adams operations
$\Psi_n$ by
$$
\Psi_n f = f(g^n) \Big|_{p_k\mapsto p_{kn}\,, \forall k} \,.
$$

\begin{lem}
  \begin{equation}
    \label{fsSd}
    \sSd f = \exp \left(\sum \frac1n  \Psi_n f 
    \right) \,. 
  \end{equation}
\end{lem}

\begin{proof}
  The identification \eqref{fW} of symmetric functions with
  representation of the symmetric group may be restated as follows.
  Let $W$ be a representation of $S(k)$ and 
  suppose we want to evaluate $f_W$ as a symmetric polynomial
  of some variables $x_1,\dots,x_N$. This evaluation is given by
  $$
  f_W(x_1,x_2,\dots) = \tr_{V_W} \mathbf{x} \,, \qquad 
  \mathbf{x} = \begin{pmatrix}
    x_1 \\
    & x_2 \\
    && \ddots 
  \end{pmatrix}
\,,
$$
where
$$
V_W = \left((\C^N)^{\otimes k} \otimes W\right)^{S(k)} \,.
$$
It follows that
$$
\Psi_n f_W = \tr_{V_W} \mathbf{x}^n \,. 
$$
The diagram \eqref{SVSV} thus reduces \eqref{fsSd} to \eqref{Sdg} \,. 
\end{proof}

\subsubsection{}
Given a collection $\cF=\{\cF_k\}$ of 
$S(k)\times \bT$-equivariant sheaves $\cF_k$ on $\C^{2k}$, the
procedure \eqref{sSdd} outputs a new collection of sheaves that
we denote $\sSd \cF$. More precisely,
\begin{equation}
\sSd  \cF  = \bigoplus_{\mathbf{k}=\{k_1,k_2, \dots\} } \Ind^{S(\sum
  k_i)}_{\Aut(\mathbf{k}) \ltimes \prod S(k_i)}
\cF_{k_1} \boxtimes \cF_{k_2} \boxtimes \dots \,, \label{sSdd2}
\end{equation}
where $\boxtimes$ denotes the exterior tensor product over the
coordinate rings of $\C^{2k_i}$. The class of \eqref{sSdd2}
in K-theory may be computed
by the formula \eqref{fsSd}.

Via the BKRH identification \eqref{BKRH}, this operation may be
transported to the Hilbert schemes of points.

\subsubsection{}
If the sheaf $\cF$ is a representation of a further group $G$ that
commutes with $S(k)\times \bT$, then so is $\sSd  \cF$ and one
should remember to apply the operations $\Psi_n$ to the elements of
$G$ in \eqref{fsSd}.

In particular, from the perspective of M-theory, the grading by the
Euler characteristic in the DT theory is the grading with respect to a
multiplicative group $\C^\times \owns q$. Therefore, we set 
\begin{equation}
\Psi_n  \, q = q^n \,.\label{Psiq}
\end{equation}

\subsection{Main result}

\subsubsection{}

\begin{thm}\label{t_main} 
  \begin{equation}
    \label{mainTh}
     \ev_* \tO_\vir = \sSd  \left(
     		-\frac{1}{\{y\}} \frac{q}{1-q/\kappa} p_1 - \frac{1}{\{y\}} \frac{q}{1-q/\kappa} \bar p_1 - \frac{q}{\{y\}\{xz\}} \frac{1-q \kappa}{1-q/\kappa} p_1 \bar p_1
     \right)		
  \end{equation}
\end{thm}

We recall that it is very important to keep in mind that this
identification includes the action of $q$ as in \eqref{Psiq}.

\section{Proof of Theorem \ref{t_main}}

\subsection{The $1$-leg vertex}

\subsubsection{}

To set up the strategy of the proof, we prove a special case of the
formula first. It will also serve as an 
auxiliary statement in the proof of the full statement.

This special case concerns curves that do not meet the
divisor $D_2$. In other words, they are in the homology
class of multiples of the curve $C_1\subset S$. In terms of
the formula \eqref{mainTh} this means taking the constant
term in the $\bar p_k$'s, thus the claim to prove is 
   \begin{equation}
   \label{onelg}
    \left.  \ev_* \tO_\vir \right|_{\textup{1-leg}} = \sSd  \left(
     		-\frac{1}{\{y\}} \frac{q}{1-q/\kappa} p_1
     \right)	\,.	
  \end{equation}

\subsubsection{}\label{s_irr_curves}

We note that there are very few reduced irreducible complete
curves in
$X$. Indeed, they all must be of the form
$$
C \times \textup{point} \subset S \times \C
$$
where $C\subset S$ is reduced and irreducible, and thus
either a smooth fiber of the blow-down map
$$
S \to \mathbb{P}^1 \times \C
$$
or one of the irreducible components of the special fiber
$C_1 \cup C_2 \subset S$. The smooth fiber is a smoothing of
$C_1 \cup C_2$.

In particular, irreducible curves that do not meet $D_2$ are all of
the  form $C_1 \times \textup{point} \subset S \times \C$. 

\subsubsection{}\label{s_stable} 

Consider the one rank 1 torus that scales the $xz$- and
$y$-axes with opposite weights, attracting and repelling
respectively. We will pair $\ev_* \tO_\vir$ with the stable
envelopes $\Stab(\lambda)$ of fixed points for this torus action.
See Chapter 9 in \cite{OkLect} for an introduction to stable
envelopes in equivariant K-theory.

By the discussion
of Section \ref{s_irr_curves} the intersections of supports
of $\ev_* \tO_\vir$ and $\Stab(\lambda)$ is proper, thus
the pairing is a series in $q$ with coefficients in
nonlocalized K-theory 
$$
K_{\bTt}(\pt) = \Z[x^{\pm 1}, y^{\pm 1}, z^{\pm 1}, \kappa, y^{1/2}] \big/
(\kappa^2 - xyz) \,. 
$$
In fact, by a judicious choice of the slope and the polarization of
stable envelopes, we may achieve a sharper result.

\subsubsection{}

We take stable envelopes $\Stab(\lambda)$ with zero slope and the
polarization defined by by the polarization \eqref{T12H}.

\begin{lem}
  With this choice of parameters, we have
  \begin{equation}
\chi((xz)^{|\lambda|^2/2} \Stab(\lambda) \otimes \ev_* \tO_\vir) \in\Z[\kappa^{\pm 1}][[q]] \,.\label{stab_ev}
\end{equation}
\end{lem}

\begin{proof}
It suffices to check that $\chi(\Stab(\lambda) \otimes \ev_*
\tO_\vir)$ remains bounded as
\begin{equation}
x^{\pm 1}, y^{\pm 1}, z^{\pm 1} \to
\infty\,, \quad xyz = \textup{constant} \,. \label{limxyz}
\end{equation}
By equivariant
localization, the result is a sum of contribution of fixed points,
each of which is a rational function of $x,y,z,\kappa$.

In this rational functions, terms that correspond to the
first and the second summands in \eqref{cTvir_sum}, after the
twist by $\det \left(\cT^{1/2}_\vir\right)^{-1}$ in \eqref{tOvir}, 
become balanced, in the sense that they stay finite in the
limit \eqref{limxyz}.

The terms corresponding to last summand in \eqref{cTvir_sum}
become balanced after pairing with $\Stab(\lambda)$ by the
weight condition in the definition of stable envelopes.

\end{proof}

\subsubsection{}
We can compute \eqref{stab_ev} by taking any specific limit
of the form \eqref{limxyz}. Note from the above proof that in
this computation we won't need to consider the contribution of
fixed points in $\Hilb(D_1)$ other than the starting point
$\lambda$. Indeed, the weights in the
restriction of $\Stab(\lambda)$ to other
fixed points satisfy strict inequalities and hence their
contribution goes to $0$ in the limit \eqref{limxyz}.

\subsubsection{}
Of all possible limits in \eqref{limxyz}, we choose
\begin{equation}
x \gg 1 \gg y \gg z \,, \quad xyz = \textup{constant} \,. \label{limref}
\end{equation}
which, in the language of \cite{NO} means that make
computations in the \emph{refined} vertex limit, with
the $y$-direction preferred.

The details of this computation will be worked out in the
full $2$-leg generality below. Here we only state the
result

\begin{prop} We have 
  \begin{equation}
    \label{stab_one_leg}
    \chi((xz)^{|\lambda|^2/2} \Stab(\lambda) ,  \ev_* \tO_\vir) = s_\lambda\left(
    p_i = - \frac{q^i}{1-q^i/\kappa^i}
    \right) \,. 
  \end{equation}
\end{prop}

\subsubsection{}

To finish the proof in the 1-leg case we need the following
statements.

\begin{prop}\label{p_stab_schur}
Upon identification with symmetric function we have
  \begin{equation}
    \label{stab_schur}
    \Stab(\lambda) = \,\, \frac{y^{n/2}}{(xz)^{\frac{n^2}{2}-n}}  \cdot s_\lambda\left(
    \frac{p_i}{1-(xz)^i}
    \right)
  \end{equation}
\end{prop}

\begin{prop}\label{p_agree1}
  Formulas \eqref{stab_one_leg} and \eqref{onelg} agree. 
\end{prop}

\subsubsection{}

\begin{proof}[Proof of Proposition \ref{p_stab_schur}]
  We follow the notations of Section \ref{s_ab}. 
  Let $\Bbbk$ be a field and consider the algebra 
  $$
  \cA_0 = \Bbbk[S(n)] \ltimes \Bbbk
  \langle a_1,\dots,a_n, b_1,\dots, b_n \rangle \big/ ([a_i,b_j] =
  \delta_{ij}) \,. 
  $$
  This is the rational Cherednik algebra with parameter $0$ and thus
  the simplest quantization of the orbifold $T^*\Bbbk^{n}/
  S(n)$. 

Assume that $p=\textup{char} \, 
\Bbbk \gg 0$, which, in particular, implies that
irreducible representations of the symmetric group $S(n)$ are indexed
by all partitions $\lambda$ of $n$. Given such representation
$W^\lambda$, we consider the corresponding Verma module over $\cA_0$
\begin{equation}
\textup{Verma}(\lambda) = \cA_0 \otimes_{\Bbbk[S(n)] \ltimes 
\Bbbk[b_1,\dots,b_n]} W^\lambda \,,  \label{Vermal}
\end{equation}
where $b_i$ act on $W^\lambda$ by zero. 

The algebra $\cA_0$ is a free module of rank $p^{2n}$ over its
subalgebra
$$
\cA_{0,p} = \Bbbk[S(n)] \ltimes \Bbbk
  [a^p_1,\dots,a^p_n, b^p_1,\dots, b^p_n]  \subset \cA_0 
$$
and as a $\cA_{0,p}$-module, the Verma module \eqref{Vermal} is
$p^{n}$ 
copies of $W^\lambda \otimes \Fr^* \cO_L$, where $L$ is the 
Lagrangian subvariety 
$$
L = \{b_1=\dots=b_n=0\}
$$
and $\Fr$ is the Frobenius map. Thus the Bezrukavnikov-Kaledin
equivalence \cite{BezrukavnikovKaledin} 
\begin{equation}
D^b \cA_0-\textup{mod}  \xrightarrow{\,\, \sim \,\,}
D^b \Coh_{S(n)} \left(\C^{2n}\right)^{(1)}\label{BK2}
\end{equation}
sends Verma modules to modules of the form $W_\lambda \otimes \cO_L$,
the K-theory class of which was discussed earlier in \eqref{e_plet}.
The twist by 1 in the $\left(\C^{2n}\right)^{(1)}$ term in \eqref{BK2}
denotes the Frobenius twist.

The spherical subalgebra in $\cA_{0}$ is a quantization of
$\Hilb(\C^2,n)$ for zero value of the quantization parameter, thus the
Bridgland-King-Reid-Haiman equivalence \eqref{BKRH} is also an example
of a Bezrukavnikov-Kaledin equivalence at zero slope. It is known
\cite{BezrukavnikovOkounkov} that this equivalence sends Verma modules to suitably
normalized stable envelopes, in particular
\begin{equation}
    \label{stab_schur2}
    \Stab(\lambda) = \textup{monomial weight} \cdot \,   s_\lambda\left(
    \frac{p_i}{1-(xz)^{-i}} 
    \right) \,. 
  \end{equation}
 The prefactor may be found from comparing the leading monomials, that
 is, the restrictions to the fixed point indexed by $\lambda$. This
 concludes
 the proof.
\end{proof}

\bigskip


\subsubsection{}

\begin{proof}[Proof of Proposition \ref{p_agree1}]
  
\begin{multline}
\chi\left(y^{n/2} (xz)^n \cdot s_\lambda\left(
    \frac{p_i}{1-(xz)^i}
    \right),
 \sSd\left(
    -\frac{p_1}{\{y\}} \frac{q}{1-q/\kappa}
    \right) \right) = \\
    = \left\langle
    s_\lambda(p_i), \sSd\left(
    p_1 \frac{q}{1-q/\kappa}
    \right)
    \right \rangle_\text{Hall scalar product} = 
    s_\lambda\left|_
   { p_i = -\frac{q^i}{1-q^i/\kappa^i}}
   \right.
\end{multline}
Here we used that

\begin{multline}
\chi(f(p_i), g(p_i)) = \left. e^{
\sum_n n(1-(xz)^n (1-y^n) \frac{\p}{\p p_n} \frac{\p}{\p \bar p_n} 
}
\bar f (p_i) g(p_i) \right|_{p_i = \bar p_i = 0} = \\
 = \left. e^{
\sum_n n \frac{\p}{\p p_n} \frac{\p}{\p \bar p_n} 
}
\bar f (p_i \cdot (1-(xz)^i)) \  g(p_i\cdot(1-y^i)) \right|_{p_i = \bar p_i = 0} = \\
=
 \left \langle
\overline{f(p_i (1-1/(xz)^i))}, g(p_i (1-y^i))
\right \rangle_\text{Hall Scalar Product}.
\\
\end{multline}

\end{proof}

\noindent 
This concludes the proof of the $1$-leg case.

\subsection{Refined vertex limit}
\label{s_refined}

Here we study the limit
$$
t_1 \gg 1 \gg t_3 \gg t_2, t_1 t_2 t_3 = \kappa^{2} = \text{const}
$$
 of the equivariant PT vertex with 2 nontrivial legs along (nonpreferred) directions $t_1, t_2$. In this section we use the convention that the tangent weights at the origin are $t_i^{-1}$.
 The tangent space at a fix point $T$ decomposes as
 $$
 T = T_{>0} - T_{<0}, \  \text{and} \  T_{<0} =  \kappa^{-2} \overline{T_{<0}}.
 $$
Note that
$$
\frac{\{k^{-2}/w\}}{\{w\}} = \frac{\frac{1}{\kappa\sqrt{w}} - {\kappa\sqrt{w}}}{\sqrt{w} - \frac{1}{\sqrt{w}}} \to
\begin{cases} -\kappa^{-1}, & \mbox{if } w \to 0 \\ -\kappa, & \mbox{if } w \to \infty \end{cases},
$$
so the contribution of a fixed point is equal to
$$
(-\kappa)^{-\rk T_{>0}}
$$
\begin{prop}
Refined limit of the PT vertex with 2 legs:
\begin{equation}
(P_{\lambda \mu})_\text{refined} = \textup{prefactor} \cdot \sum_\eta  s_{\lambda/\eta} (1, q/\kappa, (q/\kappa)^2, ...)\cdot s_{\mu/\eta} (q\kappa, (q\kappa)^2, ...), \label{reflimver}
\end{equation}
\end{prop}

\begin{proof}

The tangent space to a sheaf
$$
T_3 (\cF) = \chi(\cF) + \chi(\cF, \cO) - \chi(\cF, \cF).
$$
$$
\pi_* \cF = \bigoplus_{i=0}^d \cE_i t_3^i.
$$

Let us denote by $\cF'\subset \cF$ the minimal sheaf with the same outgoing cylinders as $\cF$ on $\C^3$, and by $\cE'_i \subset \cE_i$ the same for $\cE_i$ on $\C^2$, and the quotients  by $\cV_i$, so that we have exact sequences of sheaves on $\C^2$:
$$
0 \to \cE_i \to \cE'_i \to \cV_i \to 0, \ \ \ \ 
0 \to w_i \cO \to \cO \to \cE_i' \to 0.
$$

Writing contributions of $t_3$-slices, we get
\begin{multline}
T_3(\cF) - T_3(\cF') = \sum_i t_3^i \chi(\cV_i) + \sum_i t_3^{-i-1} \chi(\cV_i, \cO) - \\
- (1-t_3^{-1}) \sum_{i,j} t_3^{j-i}( \chi(\cV_i, \cV_j) + \chi(\cE_i', \cV_j) + \chi(\cV_i, \cE_j') ) = \\
= \sum_i t_3^i \chi(\cV_i) + \sum_i t_3^{-i-1} \chi(\cV_i, \cO) - \\
- (1-t_3^{-1}) \sum_{i,j} t_3^{j-i}( \chi(\cV_i, \cV_j) + \chi( \cV_j)- w_i^{-1} \chi(\cV_j) + \chi(\cV_i, \cO) - w_j \chi(\cV_i, \cO) ) = \\
= \sum_i t_3^{i-N} \chi(\cV_i) + \sum_i t_3^{N-i-1} \chi(\cV_i, \cO) 
+ (1-t_3^{-1}) \sum_{i,j} t_3^{j-i} U_{ij},
\end{multline}
where
$$
U_{ij} = - \chi(\cV_i, \cV_j) + w_i^{-1} \chi(\cV_j)  + w_j \chi(\cV_i, \cO) . 
$$
If we denote
$$
\text{character }\cV_i = w_i \overline{K_i},
$$
where $K_i$ are some torsion sheaves (which we identify with their characters) on the Hilbert scheme with tangent weights $t_1, t_2$ (dual Hilbert scheme),
we get the following expression 
$$
U_{ij} = \frac{w_i^{-1} w_j}{t_1 t_2} \left(
K_i + \overline{K_j} t_1 t_2 - K_i \overline{K_j} (1-t_1)(1-t_2)
\right).
$$
If we denote ideals
$$
I_i = \Ker(\cO \to K_i),
$$
then we can write the operator as
$$
U_{ij} = \frac{1}{t_1 t_2}\left(
\chi(w_i \cO, w_j \cO) - \chi(w_i I_i, w_j I_j).
\right)
$$
which is identical to the Carlson-Okounkov Ext-operator,
so it does not have weights $(t_1 t_2)^n$, and that's why multiplied by $1-t_3^{-1}$ get the same number of positive and negative weights.

\begin{center}
  \includegraphics{pt2d.20}
\end{center}

So, the index limit of $T_3(\cF)$ is a sum of contributions of each box, and the contribution of box $t_1^i t_2^j t_3^k$ depends on the sign of $i-j$:
\begin{equation}\label{contr1}
\text{contribution}(t_1^i t_2^j t_3^k) = \begin{cases}
-\kappa^{-1}, & i-j < 0 \\
-\kappa, & i-j \geq 0
\end{cases}.
\end{equation}
Let us define operators on the Fock space:
$$
\Gamma_+(z)  = \exp\left( \sum_{n \geq 1} \frac{z^n}{n} p_n
\right), 
\ \ \ 
\Gamma_-(z)  = \exp\left( \sum_{n \geq 1} z^n \frac{\p}{\p p_n}
\right).
$$
The refined limit of the 2-leg vertex then can be represented as
$$
P_{\lambda \mu} = \langle \lambda|
...
\Gamma_-(q^2/\kappa^2) \Gamma_-(q/\kappa)
\Gamma_-(1) \Gamma_+(q \kappa) \Gamma_+(q^2 \kappa^2) ...
|\mu\rangle,
$$
and using the properties
$$
\prod_i \Gamma_-(x_i) |\lambda\rangle = \sum_{\mu \subset \lambda}
s_{\lambda/\mu}(x_i) |\mu\rangle,
$$
$$
\prod_i \Gamma_+(x_i) |\lambda\rangle = \sum_{\mu \supset \lambda}
s_{\mu/\lambda}(x_i) |\mu\rangle,
$$
we get the statement.

Then we have obtained the formula up to a prefactor.
To figure out the prefactor, note that in the definition of vertex we have to consider
$$
T_3(\cF) - T_3(C_1) - T_3(C_2),
$$
rather then
$$
T_3(\cF) - T_3(\cF'),
$$
where $C_i$ are the 2 cylinders such that
$$
\text{support}(\cF') = C_1 \cup C_2.
$$
Their difference is
$$
T_3(\cF') - T_3(C_1) - T_3(C_2).
$$
Though there is a very easy combinatorial formula for this expression, we need only its limit, and it has already been investigated in the paper \cite{NO}. Then the prefactor is equal to a product over boxes in $C_1 \cap C_2$, and
\begin{equation}\label{contr2}
\text{contribution}(t_1^i t_2^j t_3^k) = \begin{cases}
-\kappa, & i-j < 0 \\
-\kappa^{-1}, & i-j \geq 0
\end{cases}.
\end{equation}
This has a clear interpretation: the vertex starts when we have "holes" in the union of 2 cylinders (in the intersection boxes have multiplicities 1 instead of 2), and that's why the contributions in
(\ref{contr1}) and (\ref{contr2}) are the opposite.
\end{proof}

\subsection{Factorizable sheaves}
\label{s_factorizable}

\subsubsection{}

We now consider the full two legs geometry with two evaluation
maps $\ev_i$ to the Hilbert scheme of points of two divisors
$D_1$ and $D_2$. We may view the two leg vertex as an operator
acting from $K_\bT(\Hilb(D_2))$ to $K_\bT(\Hilb(D_1))$. In
particular, consider 
  \begin{equation}
    \label{defcFl}
    \cF_\lambda = \ev_{1,*} \left(\tO_\vir \otimes \ev_2^*
       \Stab(\lambda)\right) \,. 
  \end{equation}
 The already
  established case of a $1$-leg vertex gives
  \begin{equation} 
  \cF_\varnothing = \sSd  \left(
     		-\frac{1}{\{y\}} \frac{q}{1-q/\kappa} p_1 
     \right)	\,.	\label{Fnot} 
   \end{equation}
   Since the multiplication by this symmetric function is invertible,
   we may define
   \begin{equation}
   \cG_\lambda  = \cF_\varnothing^{-1} \cdot \cF_\lambda \,, \ \ \text{and} \ \ \ \cG = \cF_\varnothing^{-1} \cdot \ev_* \tO
\label{defcGlambda}
 \end{equation}%
   where the dot denotes multiplication of symmetric functions.

   The sheaves $\cF_\lambda$ and $\cG_\lambda$ are related by the
   action of correspondences from Section \ref{s_Ed}. 

   \subsubsection{}

   Let
   $$
   \Attr \subset \Hilb(D_1)
   $$
   denote the full attracting set for the torus action as in Section
   \ref{s_stable}, that is, the set of point that have a limit as the
   $(xz)$-coordinate is scaled to $0$ while the $y$-coordinate is
   scaled to $\infty$. We have
   $$
   \Attr =\{\textup{subschemes set-theoretically supported on the
     $(xz)$-axis}\} \,.
   $$
   Our next goal is the following

   \begin{prop}
     \begin{equation}
     \supp \cG_\lambda \subset \Attr \,.
\label{suppcG}
\end{equation}
   \end{prop}

   \noindent In other words, $\cG$ comes from the first term in the
   following exact sequence
   $$
   K_\bT(\Attr) \to K_\bT(\Hilb) \to K_\bT(\Hilb\setminus \Attr) \to 0
   \,.
   $$

   \subsubsection{}
   The sheaf $\cF_\lambda$ is factorizable in the following
  sense. Consider an open subset $U$ of $\Hilb(D_1)$ formed by
  subschemes of the form $Z_1 \cup Z_2$ such that
  $$
  \supp Z_1 \cap \supp Z_2 = \varnothing
  $$
  and $\supp Z_1$ does not meet the $(xz)$-axis. Then
  \begin{equation}
  \cF_\lambda\big|_U = \pi^*\left(\cF_\varnothing \boxtimes
  \cF_\lambda\right)\label{Ffact}
\end{equation}
where
$$
\pi (Z_1 \cup Z_2) = (Z_1,Z_2) \in \Hilb(D_1,|Z_1|) \times
\Hilb(D_2,|Z_2|) \,. 
$$
Indeed, the only curves in our geometry that can leave the
zero section $S\subset X$ are curves of the class $C_1$.
Once they leave the zero section, their deformation theory is
independent of the rest of the curve and of the partition $\lambda$.

\subsubsection{}

  \begin{proof}
    By construction, 
  \begin{equation}
    \label{defcFl2
    }
    \cF_\lambda = \cF_\varnothing \cdot \cG_\lambda\,. 
  \end{equation}
  Using the factorization \eqref{Ffact}, we now prove
  \eqref{suppcG}  by induction of the number $n$ in $\Hilb(D_1,n)$.
  The statement being vacuous for $n=0$, assume that
the sheaf
$$
\cG_{\lambda,n-1} = \cG_\lambda\big|_{\sqcup_{k< n} \Hilb(D_1,k)}
$$
is supported on the set $\Attr$. Consider
$$
U = \Hilb(D_1,n) \setminus \Attr \,. 
$$
By factorization
$$
\cG_\lambda\Big|_U = \left(\cF_\lambda - \cF_\varnothing \cdot \cG_{\lambda,n-1}
\right)\Big|_U = 0 \in K_\bT(U) \,, 
$$
as was to be shown. 
\end{proof}

\subsection{Conclusion of the proof}

Let's pair $\cG_\mu$ with a class of the stable envelope correspinding to the opposite chamber.  Similarly as in the 1-leg case, we have
\begin{lem}
$$
\chi(\Stab(\lambda) \otimes \cG_\mu) \in \textup{monomial weight} \cdot  \Z[\kappa^{\pm 1}][[q]].
$$
\end{lem}

Thus this function can be extracted from the refined limit (\ref{limref}), where $y$-direction is preferred. To get the formulas simpler, we should divide it by the contribution of the gluing operator ${\bf G}: K_{eq}(\Hilb(D_2)) \to K_{eq}(\Hilb(D_2))$,
which is known by \cite{OkSmir}:
$$
{\bf G} = \sSd \left(
\frac{1-q\kappa}{1-q/\kappa} p_1 \bar p_1
\right)
$$

\begin{lem}
\begin{equation}
\textup{monomial weight} \cdot \chi(\Stab(\lambda) \otimes {\bf G}^{-1}(\Stab(\mu)),
\ \ev_* \cG 
) = \left.|\kappa|^{\lambda} s_{\mu/\lambda}\right|_{p_i = -\frac{q^i}{1-(q\kappa)^i}} \label{2ls}
\end{equation}
\end{lem}

\begin{proof}
By localization,
$$
\ev_* \tO_\vir = |\text{capping}^{0,-1}_{D_1}|^t  \cdot \text{vertex}  \cdot |\text{capping}^{0,-1}_{D_2}|,
$$
where $\text{capping}$ is the solution to a certain $q$-difference equation as in \cite{OkLect}. In the limit we are considering
$$
\lim  \ \langle \Stab(\lambda) \  |\text{capping}^{0,-1}_{D_1}|  \ \text{fixed}(\mu) \rangle = \textup{monomial weight}(\lambda) \cdot \delta_{\lambda^t \mu},
$$
$$
\lim  \ \langle \Stab(\lambda) {\bf G}^{-1} \  |\text{capping}^{0,-1}_{D_2}|  \ \text{fixed}(\mu) \rangle = \text{monomial weight}(\lambda) \cdot \delta_{\lambda^t \mu},
$$
which basically means that the solution of a $q$-difference equation is trivial in the limit $q \to 0$, but we have to renormalize the basis appropriately, that is why we consider the matrix element between a stable envelope and a torus-fixed point. The prefactors here are monomials and exactly compensate the ones in (\ref{reflimver}).
\end{proof}

\begin{lem}
Formulas (\ref{mainTh}) and (\ref{2ls}) are equivalent.
\end{lem}

\begin{proof}
\begin{multline}
\chi\left(
y^{-|\lambda|/2} s_\lambda\left(\frac{p_i}{1-y^i}\right) \otimes  s_\mu\left(\frac{\bar p_1^i}{1-(xz)^{-i}}\right), \right.
\\
\left.
\sSd \left(
-\frac{1}{\{y\}} \frac{q}{1-q \kappa} \bar p_1 - \frac{q}{\{y\}\{xz\}} p_1 \bar p_1
\right)
\right) = \\
=
\left \langle
y^{-|\lambda|/2} s_\lambda(-y^i p_i) \otimes  s_\mu(\bar p_i),
\sSd \left(
\sqrt{y} \frac{q}{1-q \kappa} \bar p_1 - \kappa p_1 \bar p_1
\right)
\right \rangle_\text{Hall scalar product} = \\
=
\left \langle
s_\lambda(p_i) \otimes s_\mu(\bar p_i), \sSd \left(
-\frac{q}{1-q \kappa} \bar p_1  + \kappa p_1 \bar p_1
\right)
\right \rangle = 
\left.|\kappa|^{\lambda} s_{\mu/\lambda}\right|_{p_i = -\frac{q^i}{1-(q\kappa)^i}}
\end{multline}
\end{proof}

\section{Operator corresponding to resolved conifold (4-point
  function)}

%

In this section we will apply the computation of the two-legged vertex 
for the following relative geometry $X/D$:

\begin{center}
  \includegraphics{pic.1} \label{piccon}
\end{center}

This 3-fold has $H_2(X, \Z) = \Z$, and we define the virtual structure sheaf
$$
\tO_\vir = \cO_\vir \otimes \left(
\cK \otimes \det \text{H}^\bullet(\cF|_D)
\right)^{1/2} \cdot (-q)^\chi \otimes Q^\text{deg}.
$$
Let 
$$
\ev: \PT(X/D) \to \Hilb(D)
$$
be the restriction map.
As in the previous section, we identify
$$
K_{eq}(\Hilb(D)) = K_{eq}(\text{pt}) \otimes \Lambda(p) \otimes \Lambda(\bar p).
$$

\begin{thm}
\begin{multline*}
\ev_*{\cal \tO}^\text{vir} \cdot
\left(\left.\ev_*{\cal \tO}^\text{vir}\right|_{Q=0}\right)^{-1} = \sSd \left(
-\frac{q\left( 1 + \frac{\{y\}}{\{xz\}} Q \right)}{\{y\}(1-q/\kappa)} p_1 
- \right. \\
-\left.\frac{q\left(
1 + Q \frac{\{xz\}}{\{y\}}
\right)}{\{xz\} (1-q/\kappa)} \bar p_1 - \frac{Q q}{\{y\}\{xz\}} 
\frac{1-q\kappa}{1-q/\kappa} p_1 \bar p_1
\right)
\end{multline*}
\end{thm}

\begin{proof}
The result is obtained by gluing two capped 2-leg vertices along one leg, and can be represented pictorially.

%
%
%


\begin{center}
  \includegraphics{fig10.0}
\end{center}

The operator is the composition of the following operators:
\be
V_1 = \sSd \left(
-\frac{q}{\{y\}(1-q/\kappa)} {\bf p}_1 - \frac{q}{\{y\}(1-q/\kappa)} {\bf q}_1 - \frac{q}{\{y\}\{xz\}} \frac{1-q \kappa}{1-q/\kappa} {\bf p}_1 {\bf q}_1
\right),
\ee
\be
V_2 = \sSd \left(
-\frac{q}{\{xz\}(1-q/\kappa)} {\bf r}_1 - \frac{q}{\{xz\}(1-q/\kappa)} {\bf s}_1 - \frac{q}{\{y\}\{xz\}} \frac{1-q \kappa}{1-q/\kappa} {\bf r}_1 {\bf s}_1
\right),
\ee
\be
G^{-1} = \sSd \left(
-\frac{Q}{q} \frac{1-q/\kappa}{1-q \kappa} \{xz\}\{y\} {\bf q}_1 {\bf r}_1
\right).
\ee
They act on the following Fock spaces:
\be
V_1 : \text{Fock}({\bf p}) \to \text{Fock}({\bf q}) \\
G^{-1} : \text{Fock}({\bf q}) \to \text{Fock}({\bf r}) \\
V_2 : \text{Fock}({\bf r}) \to \text{Fock}({\bf s})
\ee
Here we introduce $Q$ into the gluing operator in order to measure degrees of curves.
We contract bilinear forms using the standard Hall product on symmetric functions.

Taking the composition of the first vertex operator with the gluing matrix, we obtain,
\be
V_1 \cdot G^{-1} = \sSd \left(
-\frac{q}{\{y\}(1-q/\kappa)} {\bf p}_1 + Q \frac{1}{1-q \kappa} \{xz\}{\bf r}_1 + Q {\bf p}_1 {\bf r}_1
\right).
\ee
Finally, composing with the second vertex operator, we obtain 6 terms from contractions:
\begin{multline*}
\sSd \left(
 - \frac{q Q}{(1-q/\kappa)(1-q \kappa)} -
\frac{q}{\{y\}(1-q/\kappa)} {\bf p}_1 
  - \frac{q Q}{\{y\}(1-q/\kappa)} {\bf s}_1 +\right. \\
\left.   - \frac{q Q}{\{xz\}(1-q/\kappa)} {\bf p}_1 -
\frac{q}{\{xz\}(1-q/\kappa)} {\bf s}_1 - \frac{q Q}{\{y\}\{xz\}} \frac{1-q \kappa}{1-q/\kappa} {\bf p}_1 {\bf s}_1
\right).
\end{multline*}
The first term in the plethystic exponential is the partition function of ${\cal O}(-1) \oplus {\cal O}(-1) \to \mathbb{P}^1$, 
$$
\mathsf{Z}_\text{resolved conifold} = \bSd \left(
 - \frac{q Q}{(1-q/\kappa)(1-q \kappa)} \right),
$$
which factors out completely,  concluding the proof.

\end{proof}


\begin{thebibliography}{9} 

\bibitem{BezrukavnikovKaledin} R.~Bezrukavnikov, D.~Kaledin, \emph{Fedosov quantization in positive characteristic}, 	arXiv: math/0501247,


\bibitem{BezrukavnikovOkounkov} R.~Bezrukavnikov, A.~Okounkov, \emph{Monodromy and derived equivalences}, in preparation,


\bibitem{BKR} T.Bridgeland, A.King, M.Reid, \emph{Mukai implies McKay:
the McKay correspondence as an equivalence of derived categories}, arXiv:math/9908027,

\bibitem{FTs} B.Feigin A.Tsymbaliuk, \emph{Equivariant K-theory of Hilbert schemes via shuffle algebra}, arXiv:0904.1679,


\bibitem{GopVaf} R.Gopakumar, C.Vafa, \emph{M-Theory and Topological strings-I, II}, arXiv:hep-th/9809187,  arXiv:hep-th/9812127,

\bibitem{Haiman} M.~Hailman, \emph{Hilbert schemes, polygraphs, and the Macdonald positivity conjecture}, arXiv: math/0010246,

\bibitem{IKV} A.Iqbal, C.Kozcaz, C.Vafa, \emph{The Refined Topological Vertex}, arXiv:hep-th/0701156,

\bibitem{INOV} A. Iqbal, N. Nekrasov, A. Okounkov, C. Vafa,
  \emph{Quantum Foam and Topological Strings}, arXiv:hep-th/0312022,


\bibitem{MNOP} D.Maulik, N.Nekrasov, A.Okounkov, R.Pandharipande, \emph{Gromov-Witten theory and Donaldson-Thomas theory, I, II}, 
 arXiv:math/0312059, arXiv:math/0406092,

\bibitem{MOOP} D.Maulik, A.Oblomkov, A.Okounkov, R.Pandharipande, \emph{ Gromov-Witten/Donaldson-Thomas correspondence for toric 3-folds}, arXiv:0809.3976,

\bibitem{Neg} A.Negut, \emph{Moduli of flagged sheaves and their K-theory}, arXiv:1209.4242v4,

\bibitem{Negut} A.~Negut,
\emph{Moduli of Flags of Sheaves and their K-theory
}, arXiv: 1209.4242,

\bibitem{NO} N.Nekrasov, A.Okounkov, \emph{Membranes and Sheaves}, arXiv:1404.2323,

\bibitem{OkLect} A.Okounkov, \emph{Lectures on K-theoretic computations in enumerative geometry}, 
arXiv:1512.07363,

\bibitem{OkTak} A.Okounkov, \emph{Takagi lectures on Donaldson-Thomas theory}, 
arXiv:1802.00779,

\bibitem{OkSmir} A.Okounkov, A.Smirnov, \emph{Quantum difference equation for Nakajima varieties}, 
arXiv:1602.09007,


\bibitem{OP} A.Okounkov, R.Pandharipande, \emph{The local Donaldson-Thomas theory of curves}, arXiv:math/0512573,





\bibitem{PandharipandeThomas} R.~Pandharipande, R.P.~Thomas, \emph{Curve counting via stable pairs in the derived category}, arXiv: arXiv:0707.2348,







\bibitem{Zel} A.~Zelevinsky, 
\emph{Representations of finite classical groups. 
  A Hopf algebra approach}, Lecture Notes in Mathematics,
869, Springer-Verlag, Berlin-New York, 1981. 




\end{thebibliography}
\end{document}